\documentclass[letterpaper, 10 pt, conference]{ieeeconf}  

\IEEEoverridecommandlockouts                              
\usepackage{color}
\usepackage[T1]{fontenc}
\usepackage{amsthm}
\usepackage{amsmath}
\usepackage{amssymb}
\usepackage{graphicx}

\usepackage{calc}

\usepackage{setspace}

\usepackage[unicode=true,
 bookmarks=true,bookmarksnumbered=true,bookmarksopen=true,bookmarksopenlevel=1,
 breaklinks=false,pdfborder={0 0 0},backref=false,colorlinks=false]
 {hyperref}
\hypersetup{pdftitle={Your Title},
 pdfauthor={Your Name},
 pdfpagelayout=OneColumn, pdfnewwindow=true, pdfstartview=XYZ, plainpages=false}

\makeatletter
\theoremstyle{plain}
\newtheorem{thm}{\protect\theoremname}
\theoremstyle{remark}
\newtheorem{rem}[thm]{\protect\remarkname}
\theoremstyle{plain}
\newtheorem{lem}[thm]{\protect\lemmaname}

\usepackage[caption=false,font=footnotesize]{subfig}

\makeatother

\providecommand{\lemmaname}{Lemma}
\providecommand{\remarkname}{Remark}
\providecommand{\theoremname}{Theorem}

\begin{document}
\global\long\def\argmax#1{\underset{#1}{\operatorname{argmax}}}

\title{\LARGE \bf
Concentration to Zero Bit-Error Probability for Regular LDPC Codes
on the Binary Symmetric Channel: Proof by Loop Calculus
}

\author{Marc Vuffray$^{1}$ and Theodor Misiakiewicz$^{2}$
\thanks{$^{1}$M. Vuffray is affiliated to the Center for Nonlinear Studies and Theoretical Division T-4 of Los Alamos National Laboratory, Los Alamos, NM 87544.
        {\tt\small vuffray@lanl.gov}}%
\thanks{$^{2}$Theodor Misiakiewicz is affiliated to ICFP, Département de Physique, Ecole Normale Supérieure de Paris,
75005 Paris, France
        {\tt\small theodor.misiakiewicz@ens.fr}}%
}

\maketitle
\thispagestyle{empty}
\pagestyle{empty}

\begin{abstract}

In this paper we consider regular low-density parity-check codes over
a binary-symmetric channel in the decoding regime. We prove that up
to a certain noise threshold the bit-error probability of the bit-sampling
decoder converges in mean to zero over the code ensemble and the channel
realizations. To arrive at this result we show that the bit-error
probability of the sampling decoder is equal to the derivative of
a Bethe free entropy. The method that we developed is new and is based
on convexity of the free entropy and loop calculus. Convexity is needed
to exchange limit and derivative and the loop series enables us to
express the difference between the bit-error probability and the Bethe
free entropy. We control the loop series using combinatorial techniques
and a first moment method. We stress that our method is versatile
and we believe that it can be generalized for LDPC codes with general
degree distributions and for asymmetric channels.

\end{abstract}

\section{Introduction\label{sec:Introduction}}

In 1968 Gallager \cite{gallager1968information} introduced error-correcting
codes based on low-density parity-check (LDPC) matrices. Since then
LDPC codes have been proven to be of great practical and theoretical
relevance. LDPC codes perform very well under iterative decoding on
a broad class of symmetric memoryless channels (BMS) \cite{Urb.45,Richardson2001}
and provably achieve capacity on the binary erasure channel (BEC)
\cite{oswald2002capacity}. Since 1996 they have been integrated into
many industrial standards from wireless communications to computer
chips.

An important performance measure of an LDPC code and its associated
decoder is the bit-error probability. It is the fraction of bits that
are on average incorrectly reconstructed. The bit-error probability
of LDPC codes under belief-propagation (BP) decoding is well-understood
on BMS channels using the method of density evolution \cite{Richardson2001a}.
However it is a more challenging task to control the bit-error probability
of the bit maximum a posteriori (MAP) decoder. 

Lower and upper bounds on the noise threshold for vanishing bit-MAP
error probability have already been derived in Gallager's thesis \cite{gallager1968information}
for a class of BMS channels. These bounds have been improved and generalized
for every BMS channels by Shamai and Sason \cite{Shamai2002}.

In an attempt to locate exactly the noise threshold, most of the attention has been focused on the conditional entropy per bit. Its derivative
with respect to the channel noise is the so-called generalized extrinsic
information transfer (GEXIT) curve \cite{Measson2009}. The GEXIT
curve is proportional to the ``magnetization'' or bit-error probability
of the bit-sampling decoder\footnote{The bit-sampling decoder assigns random values to decoded bits based on their
posterior marginal distribution.} for the BEC and the binary-input additive white Gaussian-noise channel
(BAWGNC) \cite{Korada2008}. The magnetization is an upper-bound on
the bit-MAP error probability. Hence for these two channels GEXIT
curves and bit-MAP error probabilities vanish in the same noise regime.

Surprisingly the conditional entropy and its derivative are related
to the BP algorithm and its associated Bethe free entropy. It has
been first proven in \cite{richardson2008modern,korada2007exact}
on the BEC channel that the conditional entropy is equal to an averaged
form of the Bethe free entropy over the code ensemble. Bounds between
the averaged Bethe free entropy and the conditional entropy are derived
in \cite{montanari2005tight}, \cite{macris2007griffith}, \cite{kudekar2009sharp}
based on the interpolation method of Guerra and Toninelli \cite{guerra2004high,guerra2002interpolation}.
Equality has been proven on the binary-symmetric channel (BSC) using
cluster expansions in a low-noise regime \cite{kudekar2011decay}
and in a high-noise regime \cite{Macris2013}. More recently equality
between the GEXIT curve and the derivative of the average Bethe free
entropy has been generalized to all BMS channels \cite{Giurgiu2013}
combining the interpolation method and spatially-coupled codes \cite{ZigFel}. 

Although the conditional entropy and its threshold are completely
characterized for BMS channels, its exact relation to the bit-MAP
error probability remains unclear in general and in particular for
the BSC channel. Due to Fano's inequality the conditional entropy
is always a lower-bound on the bit-error probability. However inspired
by previous results this inequality is conjectured to be tight for
LDPC codes on a wide class of channels. 

In this paper we prove that for regular LDPC codes over a BSC channel
the ``magnetization'' or bit-error probability of the bit-sampling
decoder vanishes up to a certain threshold. This result also shows
that the posterior measure of LDPC codes concentrates over the LDPC
ensemble and the noise realizations. To achieve this result we show
that the magnetization is asymptotically equal to a perturbed version
of the Bethe free entropy. The technique that we present is new and
is based on loop calculus or loop series derived by Chertkov and Chernyak
\cite{chertkov2006loop}. The loop series expresses the difference
between a quantity and its Bethe counterpart as a sum over subgraphs.
Proving that the loop series vanishes is tantamount to controlling a purely
combinatorial object that depends solely on the LDPC graph ensemble.
Suboptimal bounds on this object are obtained using McKay's estimates
\cite{mckay2010subgraphs} following an idea developed in \cite{Macris2013,Macris2012}.

The technique that we present has the advantage to be simple and versatile.
To emphasize this point we also show that our results can be easily
transposed to the BEC. Moreover we stress that our proofs do not rely
explicitly on properties of the channel. Hence, we believe that this
technique can be use to analyze LDPC codes over channels that are
not symmetric.

In Section \ref{sec:Main-Results} we give a precise definition of
the bit-sampling decoder and its associated bit-error probability
and we present our main theorems. In Section \ref{sec:Free-Entropy-Bethe}
we derive the relation between the Bethe free entropy and the bit-error
probability and we express the difference using loop calculus. In
Section \ref{sec:Free-Entropy-Bethe} we reduce the loop series to
a counting problem that we control with a first moment method. Finally
we discuss about future works and possible improvements in Section
\ref{sec:Path-Forward}.

\section{Main Results\label{sec:Main-Results}}

\subsection{Regular LDPC codes on BMS channels}

LDPC codes are defined by a regular bipartite graph $\Gamma=\left(V,C,E\right)$
where $V$ is the set of variable nodes, $C$ is the set of check nodes
and $E=V\times C$ is the set of undirected edges. There are $n=\left|V\right|$
variable nodes and $m=\left|C\right|$ check nodes. 

We consider regular LDPC codes with variable-node degrees $l\geq3$
and check-node degrees $r>l$. The design rate of the code is by definition
$R_{\text{des}}=1-l/r.$ 

An LDPC code is generated randomly. The graph $\Gamma$ is drawn uniformly
at random from the ensemble of $\left(l,r\right)$ regular bipartite
graphs. Throughout the paper we write $\mathbb{E}_{\Gamma}\left[\cdot\right]$
the expectation with respect to the ensemble of regular $\left(l,r\right)$
bipartite graphs with uniform probability.

Denote the neighbors of a variable node $i\in V$ (resp. a check node
$a\in C$) by $\partial i=\left\{ a\in C\mid\left(i,a\right)\in E\right\} $
(resp. by $\partial a=\left\{ i\in V\mid\left(i,a\right)\in E\right\} $).
A codeword is a sequence\footnote{As we use concepts from statistical physics it is more convenient
to employ the binary alphabet $\left\{ -1,1\right\} $ instead of
the traditional $\left\{ 0,1\right\} $. } $\underline{\sigma}=\left\{ \sigma_{i}\right\} _{i=1}^{n}\in\left\{ -1,1\right\} ^{n}$
that satisfies the parity-check sum
\begin{equation}
\prod_{i\in\partial a}\sigma_{i}=1,\label{eq:pr_partity_check}
\end{equation}
for all check nodes $a\in C$.

We transmit a codeword with uniform prior over a BMS channel with
transition probability $q\left(s_{i}\mid\sigma_{i}\right)$, where
the output of the channel could take any real value $s_{i}\in\mathbb{R}$. The symmetry property of the channel is expressed through the simple
relation
\begin{equation}
q\left(s_{i}\mid\sigma_{i}\right)=q\left(-s_{i}\mid-\sigma_{i}\right).\label{eq:pr_channel_symmetry}
\end{equation}
We assume without loss of generality that the all-zero codeword\footnote{In the binary alphabet $\left\{ -1,1\right\} $, the all-zero codeword
is the sequence $\left\{ 1,\ldots,1\right\} $.} is transmitted. Hence the output of the channel $\underline{s}=\left\{ s_{i}\right\} _{i=1}^{n}\in\mathbb{R}^{n}$
is i.i.d. with distribution $q\left(s_{i}\mid+1\right)$. The posterior
probability that the codeword $\underline{\sigma}$ is sent given
that $\underline{s}$ is transmitted reads 
\begin{equation}
\mu_{\Gamma}\left(\underline{\sigma}\mid\underline{s}\right)=\frac{1}{Z\left(\Gamma,\underline{s}\right)}\prod_{a\in C}\frac{1}{2}\left(1+\prod_{i\in\partial a}\sigma_{i}\right)\prod_{i\in V}q\left(s_{i}\mid\sigma_{i}\right),\label{eq:pr_posterior_measure}
\end{equation}
where the normalization factor $Z\left(\Gamma,\underline{s}\right)$
in Equation \eqref{eq:pr_posterior_measure} is the partition function
\begin{equation}
Z\left(\Gamma,\underline{s}\right):=\sum_{\underline{\sigma} \in \{-1,1\}^n}\prod_{a\in C}\frac{1}{2}\left(1+\prod_{i\in\partial a}\sigma_{i}\right)\prod_{i\in V}q\left(s_{i}\mid\sigma_{i}\right).\label{eq:pr_partition_function}
\end{equation}

\subsection{Concentration of the Bit-Error Probability for the Sampling Decoder}

We are interested in the performance of regular LDPC codes with respect
to the average bit-error probability of decoding. We consider the
bit-sampling decoder 
\begin{equation}
\widehat{\sigma}_{i}^{\text{sampling}}\left(\underline{s}\right):=\text{sample}\,\sigma_{i}\,\text{according to}\,\sum_{\underline{\sigma}\setminus\sigma_{i}}\mu\left(\underline{\sigma}\mid\underline{s}\right),\label{eq:pr_sampling_decoder}
\end{equation}
where $\underline{\sigma}\setminus\sigma_{i}$ denotes the sequence of variables $\underline{\sigma}$ with the $i$\textsuperscript{th} component removed.

The bit-error probability of the bit-sampling decoder $P_{\Gamma}^{\text{bit-sampling}}$ is directly
related to the marginals of the posterior probability \eqref{eq:pr_posterior_measure}
\begin{eqnarray}
P_{\Gamma}^{\text{bit-sampling}} & := & \frac{1}{2}\left(1-\mathbb{E}_{\underline{s}}\left[\frac{1}{n}\sum_{i=1}^{n}\left\langle \sigma_{i}\right\rangle _{\mid\underline{s}}\right]\right),\label{eq:pr_average_magnetization}
\end{eqnarray}
where $\mathbb{E}_{\underline{s}}\left[\cdot\right]$ denotes the
expectation with respect to the channel output distribution and $\left\langle \cdot\right\rangle _{\mid\underline{s}}$
denotes the average with respect to the posterior probability \eqref{eq:pr_posterior_measure}.
The expected quantity in Equation \eqref{eq:pr_average_magnetization}
is sometimes referred as the averaged magnetization in the physics
community.

An important question is to know when the bit-error probability is
vanishing in the limit where the codeword length goes to infinity.
In this paper we consider two families of symmetric channels, the
BEC and the BSC. The BEC has an output alphabet $s_{i}\in\left\{ -1,0,1\right\} $
and is characterized by transition probabilities
\begin{equation}
q^{\text{BEC}}\left(1\mid1\right)=1-\epsilon,\,q^{\text{BEC}}\left(0\mid1\right)=\epsilon,\,q^{\text{BEC}}\left(-1\mid1\right)=0,
\end{equation}
where $\epsilon\in\left[0,1\right]$ is the erasure probability. The
BSC has binary outputs $s_{i}\in\left\{ -1,1\right\} $ and is characterized
by the transition probabilities
\begin{equation}
q^{\text{BSC}}\left(1\mid1\right)=1-p,\,q^{\text{BSC}}\left(-1\mid1\right)=p,
\end{equation}
where $p\in\left[0,1/2\right]$ is the flipping probability.

Before we state our theorems we need to introduce the domain
\begin{eqnarray}
D\left(\rho\right) & = & \left\{ \left(x_{0},x_{c},\underline{y}\right)\in\left[0,1\right]^{2+\left\lfloor r/2\right\rfloor }\mid\sum_{t=1}^{\left\lfloor r/2\right\rfloor }y_{t}\leq1,\right.\nonumber \\
 &  & \left.\sum_{t=1}^{\left\lfloor r/2\right\rfloor }\frac{2t}{r}y_{t}=\left(1-\rho\right)x_{0}+\rho x_{c}\right\} .\label{eq:pr_domain}
\end{eqnarray}
We also need to introduce the auxiliary function $f:D\left(\rho\right)\times\left[0,1\right]\rightarrow\mathbb{R}$
defined as follows\footnote{The binary entropy $h_{2}\left(p\right):=-\left(1-p\right)\ln\left(1-p\right)-p\ln p$
is computed in nat.} 
\begin{eqnarray}
f\left(x_{0},x_{c},\underline{y},\rho\right) & = & -lh_{2}\left(\left(1-\rho\right)x_{0}+\rho x_{c}\right)\nonumber \\
 & &+  \left(1-\rho\right)h_{2}\left(x_{0}\right)+\rho h_{2}\left(x_{c}\right)\nonumber \\
 & &-  \frac{l}{r}\left(1-\sum_{t=1}^{r}y_{t}\right)\ln\left(1-\sum_{t=1}^{r}y_{t}\right)\nonumber \\
 & &-  \frac{l}{r}\sum_{t=1}^{r}y_{t}\ln y_{t}\nonumber \\
 & &+  \frac{l}{r}\sum_{t=1}^{\left\lfloor r/2\right\rfloor }y_{t}\ln\left(\begin{array}{c}
r\\
2t
\end{array}\right),\label{eq:pr_f_function}
\end{eqnarray}
and the function $k:\left[0,1\right]^{4}\rightarrow\mathbb{R}$ that
reads
\begin{equation}
k\left(x_{0},x_{c},\rho,p\right)=\left(\rho x_{c}-\left(1-\rho\right)x_{0}\right)\ln\left(\frac{1-p}{p}\right).\label{eq:pr_k_function}
\end{equation}

Our main contribution are the two theorems stated below which give
sufficient conditions on the channel parameters for concentration
of the bit-error probability of the sampling decoder. 
\begin{thm}[Concentration of the Bit-Error Probability for the BEC]
\label{thm:PR_concentration_error_BEC}Consider the ensemble of $\left(l,r\right)$
regular LDPC codes on a BEC with erasure probability $\epsilon$.
If the following function achieves its maximum only at the point 
\[
\argmax{\left(0,x_{c},\underline{y}\right)\in D\left(\epsilon\right)}f\left(0,x_{c},\underline{y},\epsilon\right)=\left\{ \left(0,0,0\right)\right\} ,
\]
then the bit-error probability of the sampling decoder converges in
mean to zero in the large codeword limit 
\[
\lim_{n\rightarrow\infty}\mathbb{E}_{\Gamma,\underline{s}}\left[P_{\Gamma}^{\text{bit-sampling}}\right]=0.
\]

\end{thm}
The same theorem holds for the BSC with a similar condition. 
\begin{thm}[Concentration of the Bit-Error Probability for the BSC]
\label{thm:PR_concentration_error_BSC}Consider the ensemble of $\left(l,r\right)$
regular LDPC codes on a BSC with flipping probability $p$. If the
following function achieves its maximum only at the point
\[
\argmax{\left(x_{0},x_{c},\underline{y}\right)\in D\left(p\right)}f\left(x_{0},x_{c},\underline{y},p\right)+k\left(x_{0},x_{c},p,p\right)=\left\{ \left(0,0,0\right)\right\} ,
\]
then the bit-error probability of the sampling decoder converges in
mean to zero in the large codeword limit 
\[
\lim_{n\rightarrow\infty}\mathbb{E}_{\Gamma,\underline{s}}\left[P_{\Gamma}^{\text{bit-sampling}}\right]=0.
\]
\end{thm}
\begin{rem}
Knowing that $P_{\Gamma}^{\text{bit-sampling}}$ vanishes implies
that with high probability the posterior measure \eqref{eq:pr_posterior_measure}
concentrates on configurations that are at a Hamming distance $o\left(n\right)$
from the all-zero codeword.

We perform the global optimization numerically and we find for a few
cases the maximum value of noise $\epsilon_{\text{loop}}$ and $p_{\text{loop}}$
for which Theorem~\ref{thm:PR_concentration_error_BEC} and Theorem~\ref{thm:PR_concentration_error_BSC}
hold. Critical values of noise are displayed in Table~\ref{tab:thresholds_BEC}
for the BEC and in Table~\ref{tab:thresholds_BSC} for the BSC.

\begin{table}[tbh]
\begin{centering}
\begin{tabular}{ccccccc}
\hline 
$l$ & $r$ & $R_{\text{des}}$ & $\epsilon_{\text{BP}}$ & $\epsilon_{\text{loop}}$ & $\epsilon_{\text{MAP}}$ & $\epsilon_{\text{Sh}}$\tabularnewline
\hline 
\hline 
3 & 4 & 1/4 & 0.64743 & 0.7442(9) & 0.74601 & 0.75\tabularnewline
\hline 
3 & 5 & 2/5 & 0.51757 & 0.5872(4) & 0.59098 & 0.6\tabularnewline
\hline 
3 & 6 & 1/2 & 0.42944 & 0.4833(6) & 0.48815 & 0.5\tabularnewline
\hline 
4 & 6 & 1/3 & 0.50613 & 0.5767(2) & 0.66565 & 0.66667\tabularnewline
\hline 
\end{tabular}
\par\end{centering}

\centering{}\protect\caption{\label{tab:thresholds_BEC}Thresholds for some regular LDPC code ensembles
over the BEC with erasure probability $\epsilon$. The belief-propagation
threshold is $\epsilon_{\text{BP}}$, the maximum a posteriori threshold
is $\epsilon_{\text{MAP}}$, the Shannon threshold is $\epsilon_{\text{Sh}}$
and our threshold is $\epsilon_{\text{loop}}$. Values of BP and MAP
thresholds are from \cite{mezard09information}.}
\end{table}

\begin{table}[tbh]

\begin{centering}
\begin{tabular}{ccccccc}
\hline 
$l$ & $r$ & $R_{\text{des}}$ & $p_{\text{BP}}$ & $p_{\text{loop}}$ & $p_{\text{MAP}}$ & $p_{\text{Sh}}$\tabularnewline
\hline 
\hline 
3 & 4 & 1/4 & 0.16692 & 0.2014(2) & 0.21011 & 0.21450\tabularnewline
\hline 
3 & 5 & 2/5 & 0.11382 & 0.1146(8) & 0.13841 & 0.14610\tabularnewline
\hline 
3 & 6 & 1/2 & 0.08402 & 0.0678(9) & 0.10101 & 0.11003\tabularnewline
\hline 
4 & 6 & 1/3 & 0.11692 & 0.1705(2) & 0.17261 & 0.17395\tabularnewline
\hline 
\end{tabular}
\par\end{centering}

\centering{}\protect\caption{\label{tab:thresholds_BSC}Thresholds for some regular LDPC code ensembles
over the BSC with erasure probability $p$. The belief-propagation
threshold is $p_{\text{BP}}$, the maximum a posteriori threshold
is $p_{\text{MAP}}$, the Shannon threshold is $p_{\text{Sh}}$ and
our threshold is $p_{\text{loop}}$. Values of BP and MAP thresholds
are from \cite{mezard09information}.}
\end{table}
We would expect that the probability of error vanishes
for $\epsilon<\epsilon_{\text{MAP}}$ and $p<p_{\text{MAP}}.$ Although
the thresholds that we found are reasonably close to $\epsilon_{\text{MAP}}$
and $p_{\text{MAP}}$ for graphs with small degrees, they become worse
in the limit of large degrees. A quick inspection of \eqref{eq:pr_f_function}
and \eqref{eq:pr_k_function} shows that the functions $f/l$ and
$k/l$ become independent of the noise parameter in the limit where
$l$ and $r$ go to infinity with a fixed ratio $l/r$. It implies
that $p_{\text{loop}}$ and $\epsilon_{\text{loop}}$ vanish. This
behavior is in the opposite direction to what we can expect as in
the limit of large degrees $p_{\text{MAP}}\rightarrow p_{\text{Sh}}$.
In Section~\ref{sec:Path-Forward} we discuss about possible improvements
in our analysis in order to make our thresholds tight.

The rest of the paper is organized as follows. In Section~\ref{sec:Free-Entropy-Bethe}
we show that the bit-error probability is related to the derivative
of the so-called free entropy. Using the loop series, we express the
free entropy as a combinatorial sum over subgraphs. In Section~\ref{sec:First-Moment-Method}
we control the loop series with asymptotic estimates on subgraphs
and Laplace's method. We prove Theorems~\ref{thm:PR_concentration_error_BEC}
and \ref{thm:PR_concentration_error_BSC} in this section. In Section~\ref{sec:Path-Forward}
we discuss future directions and ways to improve and generalize our
results.
\end{rem}

\section{Free Entropy, Bethe Approximation and Loop Series\label{sec:Free-Entropy-Bethe}}

\subsection{The Free Entropy and its Relation to the Bit-Error Probability}

The bit-error probability \eqref{eq:pr_average_magnetization} is
related to a ``perturbed'' version of the partition function \eqref{eq:pr_partition_function}.
Let $\eta\in\mathbb{\mathbb{R}}$ be the perturbation parameter entering
in the perturbed partition function 
\begin{eqnarray}
Z\left(\Gamma,\underline{s},\eta\right)& := &\sum_{\underline{\sigma}}\prod_{a\in C}\frac{1}{2}\left(1+\prod_{i\in\partial a}\sigma_{i}\right) \nonumber \\ 
&  &\times  \prod_{i\in V}q\left(s_{i}\mid\sigma_{i}\right) e^{\eta\left(\sigma_{i}-1\right)}.\label{eq:fe_perturbed_partition_function}
\end{eqnarray}
Note that $Z\left(\Gamma,\underline{s},\eta\right)$ is a non-increasing
function of $\eta$ and $Z\left(\Gamma,\underline{s},0\right)$ is
the original partition function \eqref{eq:pr_partition_function}.

The free entropy is the (normalized) logarithm of the partition function
\eqref{eq:fe_perturbed_partition_function}
\begin{equation}
\phi\left(\Gamma,\underline{s},\eta\right):=\frac{1}{n}\ln Z\left(\Gamma,\underline{s},\eta\right).\label{eq:fe_free_entropy}
\end{equation}
A direct computation shows that the derivative of the free entropy
with respect to its perturbation parameter reads 
\begin{equation}
\left.\frac{\partial}{\partial\eta}\phi\left(\Gamma,\underline{s},\eta\right)\right|_{\eta=0}=\frac{1}{n}\sum_{i=1}^{n}\left\langle \sigma_{i}\right\rangle _{\underline{s}}-1.
\end{equation}
Therefore the bit-error probability is related to the average entropy
through the following relation 
\begin{equation}
\left.\frac{\partial}{\partial\eta}\mathbb{E}_{\underline{s}}\left[\phi\left(\Gamma,\underline{s},\eta\right)\right]\right|_{\eta=0}=-2P_{\Gamma}^{\text{bit-sampling}}.\label{eq:fe_free_entropy_bit_error}
\end{equation}

Since $Z\left(\Gamma,\underline{s},\eta\right)$ is a non-increasing
function of $\eta$, the free entropy is non-increasing as well. Moreover
the free entropy is a convex function of $\eta$ as it can easily
be verified by taking twice the derivative with respect to $\eta$.
It implies that in order to show concentration of the bit-error probability
it is sufficient to prove that there exists $\widetilde{\eta}<0$
independent of $n$ such that $\mathbb{E}_{\Gamma,\underline{s}}\left[\phi\left(\Gamma,\underline{s},\widetilde{\eta}\right)\right]\rightarrow0$.
If this condition is true then, thanks to monotonicity, the limit
is also equal to zero for all $\eta\in\left[\widetilde{\eta},\infty\right[$.
Finally convexity of the free entropy enables us to exchange limit
and derivative (see \cite[p. 203]{Urruty2001}).

In order to prove that the free entropy vanishes we decompose it into
two contributions: the Bethe free entropy that can be computed explicitly
and the so-called loop series that is a sum over subgraphs of $\Gamma$.
Using a first moment method and combinatorial tools from graph theory,
we show that with high probability the loop series vanishes in the
large codeword limit. The last statement implies that the free entropy
is equal to the Bethe free entropy.

\subsection{The Bethe Approximation }

The Bethe free entropy is an approximation of the free entropy \eqref{eq:fe_free_entropy}.
It is defined as a functional over ``messages'' that are probability
distributions $\nu_{i\rightarrow a}\left(\sigma_{i}\right)$, $\widehat{\nu}_{a\rightarrow i}\left(\sigma_{i}\right)$
associated with the directed edges $i\rightarrow a$, $a\rightarrow i$
of the graph. The messages satisfy the so-called belief-propagation
(BP) equations. For the free entropy \eqref{eq:fe_free_entropy} the
BP equations take the following form
\begin{eqnarray}
\widehat{\nu}_{a\rightarrow i}\left(\sigma_{i}\right) & \propto & \sum_{\underline{\sigma}_{\partial a}\setminus\sigma_{i}}\frac{1}{2}\left(1+\prod_{i\in\partial a}\sigma_{i}\right)\prod_{j\in\partial a\setminus i}\nu_{j\rightarrow a}\left(\sigma_{i}\right)\nonumber \\
\nu_{i\rightarrow a}\left(\sigma_{i}\right) & \propto & e^{\eta\left(\sigma_{i}-1\right)}q\left(s_{i}\mid\sigma_{i}\right)\prod_{b\in\partial i\setminus a}\widehat{\nu}_{b\rightarrow i}\left(\sigma_{i}\right),\label{eq:be_BP_equations}
\end{eqnarray}
where the symbol $\propto$ denotes equality up to a normalization
factor and $\underline{\sigma}_{\partial a}:= \{\sigma_j \mid j\in \partial a\}$.

The Bethe free entropy evaluated at a fixed point of the BP equations
is a sum of local contributions from nodes and edges of the graph
$\Gamma=\left(V,C,E\right)$
\begin{equation}
\phi_{\left(\underline{\nu},\underline{\widehat{\nu}}\right)}^{{\rm Bethe}}\left(\Gamma,\underline{s},\eta\right):=\frac{1}{n}\sum_{a\in C}F_{a}+\frac{1}{n}\sum_{i\in V}F_{i}-\frac{1}{n}\sum_{\left(i,a\right)\in E}F_{ia},\label{eq:be_Bethe_free_entropy}
\end{equation}
where 
\begin{eqnarray}
F_{a} & = & \ln\left(\sum_{\underline{\sigma}_{\partial a}}\frac{1}{2}\left(1+\prod_{i\in\partial a}\sigma_{i}\right)\prod_{j\in\partial a}\nu_{j\rightarrow a}\left(\sigma_{i}\right)\right)\nonumber \\
F_{i} & = & \ln\left(\sum_{\sigma_{i}}e^{\eta\left(\sigma_{i}-1\right)}q\left(s_{i}\mid\sigma_{i}\right)\prod_{b\in\partial i}\widehat{\nu}_{b\rightarrow i}\left(\sigma_{i}\right)\right)\nonumber \\
F_{ia} & = & \ln\left(\sum_{\sigma_{i}}\nu_{i\rightarrow a}\left(\sigma_{i}\right)\widehat{\nu}_{a\rightarrow i}\left(\sigma_{i}\right)\right).
\end{eqnarray}
Note that once a fixed-point of the BP equations \eqref{eq:be_BP_equations}
is found, computing the Bethe free entropy \eqref{eq:be_Bethe_free_entropy}
is a computationally easy task.

\subsection{Corrections to the Bethe Free Entropy: the Loop Series}

The difference between the free entropy and the Bethe free entropy
can be expressed with the so-called loop series derived by Chertkov
and Chernyak \cite{chertkov2006loop}. It takes the form of the logarithm
of a weighted sum over subgraphs of $\Gamma$. These subgraphs are
called ``loops'' for they have no dangling edges. Note that if $\Gamma$
is a tree no such subgraph exists and we recover the well-known result
that the Bethe free entropy is exact on trees.

We recall that a subgraph $g=\left(V_{g},C_{g},E_{g}\right)$ of $\Gamma=\left(V,C,E\right)$
is any graph with vertex set $V_{g}\subset V$, factor node set $C_{g}\subset C$
and edge set $E_{g}\subset\left(V_{g}\times C_{g}\right)\cap E$.
For simplicity we denote the relation ``$g$ is a subgraph of $\Gamma$''
with the inclusion symbol $g\subset\Gamma$. We also denote the induced
neighborhood in $g$ of a variable node $i\in V_{g}$ (resp. check
node $a\in C_{g}$) by $\partial_{g}i=\partial i\cap V_{g}$ (resp.
by $\partial_{g}a=\partial a\cap C_{g}$). 

The set of ``loops'' consists of any non-empty subgraphs, not necessarily
connected, with no degree one variable-node and no degree one check-node
\begin{equation}
\mathcal{L}_{\Gamma}:=\left\{ g\subset\Gamma\mid\forall i\in V_{g},\left|\partial_{g}i\right|\geq2\text{ and }\forall a\in C_{g},\left|\partial_{g}a\right|\geq2\right\} .\label{eq:LC_loop_ensemble}
\end{equation}
The difference between the free entropy and the Bethe free entropy
is related to the loop series through the following equation
\begin{equation}
\phi\left(\Gamma,\underline{s},\eta\right)-\phi_{\left(\underline{\nu},\underline{\widehat{\nu}}\right)}^{{\rm Bethe}}\left(\Gamma,\underline{s},\eta\right)=\frac{1}{n}\ln\left(Z_{\left(\underline{\nu},\underline{\widehat{\nu}}\right)}^{{\rm loop}}\right),\label{eq:LC_loop_correction}
\end{equation}
where the argument of the logarithm is a weighted sum over loops 
\begin{equation}
Z_{\left(\underline{\nu},\underline{\widehat{\nu}}\right)}^{{\rm loop}}:=1+\sum_{g\in\mathcal{L}_{\Gamma}}K_{\left(\underline{\nu},\underline{\widehat{\nu}}\right)}\left(g\right).\label{eq:LC_Z_corr}
\end{equation}
The weight function over loops depends on the BP fixed point at which
the Bethe free entropy is evaluated and can be expressed as a product
over the nodes inside a loop 
\begin{equation}
K_{\left(\underline{\nu},\underline{\widehat{\nu}}\right)}\left(g\right):=\prod_{i\in V_{g}}\kappa_{i}\prod_{a\in C_{g}}\kappa_{a}.\label{eq:LC_loop_weight}
\end{equation}
The factors $\kappa_{i}$ and $\kappa_{a}$ entering in \eqref{eq:LC_loop_weight}
depend only on messages that are associated with edges neighboring
the nodes $i\in V_{g}$ and $a\in C_{g}$
\begin{eqnarray}
\kappa_{i} & := & \left(\sum_{\sigma_{i}}q\left(s_{i}\mid\sigma_{i}\right)e^{\eta\left(\sigma_{i}-1\right)}\prod_{a\in\partial i}\widehat{\nu}_{a\rightarrow i}\left(\sigma_{i}\right)\right)^{-1}\nonumber \\
 & \times & \left(\sum_{\sigma_{i}}q\left(s_{i}\mid\sigma_{i}\right)e^{\eta\left(\sigma_{i}-1\right)}\prod_{a\in\partial i\setminus\partial_{g}i}\widehat{\nu}_{a\rightarrow i}\left(\sigma_{i}\right)\right.\nonumber \\
 & &\times  \left.\prod_{a\in\partial_{g}i}\sigma_{i}\nu_{i\rightarrow a}\left(-\sigma_{i}\right)\right),\label{eq:LC_kappa_i}
\end{eqnarray}
and
\begin{eqnarray}
\kappa_{a} & := & \left(\sum_{\underline{\sigma}_{\partial a}}\left(1+\prod_{i\in\partial a}\sigma_{i}\right)\prod_{i\in\partial a}\nu_{i\rightarrow a}\left(\sigma_{i}\right)\right)^{-1}\nonumber \\
 & \times & \left(\sum_{\underline{\sigma}_{\partial a}}\left(1+\prod_{i\in\partial a}\sigma_{i}\right)\prod_{i\in\partial a\setminus\partial_{g}a}\nu_{i\rightarrow a}\left(\sigma_{i}\right)\right.\nonumber \\
 & &\times  \left.\prod_{i\in\partial_{g}a}\sigma_{i}\widehat{\nu}_{a\rightarrow i}\left(-\sigma_{i}\right)\right).\label{eq:LC_kappa_a}
\end{eqnarray}
For a complete derivation of the loop series for graphical models
associated with linear codes, we refer the reader to \cite{Vuffray2014}.

\subsection{The Decoding Regime and its BP Fixed-Point}

Note that the loop series, as well as the Bethe free entropy, are
functions of fixed-points of the BP equations \eqref{eq:be_BP_equations}.
The fixed-point associated with the decoding regime is the ferromagnetic
fixed-point

\begin{eqnarray}
\widehat{\nu}_{a\rightarrow i}^{+}\left(\sigma_{i}\right)= & \nu_{i\rightarrow a}^{+}\left(\sigma_{i}\right)= & \frac{1+\sigma_{i}}{2}.\label{eq:LC_Ferro_BP}
\end{eqnarray}
One can easily see that ferromagnetic messages \eqref{eq:LC_Ferro_BP}
satisfy the BP equations \eqref{eq:be_BP_equations} regardless of
the channel considered and of the value of the perturbation parameter
$\eta\in\mathbb{R}$. The ferromagnetic fixed-point \eqref{eq:LC_Ferro_BP}
describes a state for which the most likely configuration is the all-zero
codeword i.e. $\sigma_{i}=+1$. This is the reason why this fixed-point
is associated with the decoding regime.

The Bethe free entropy \eqref{eq:be_Bethe_free_entropy} evaluated
at the ferromagnetic fixed-point simply reads 
\begin{equation}
\phi_{+}^{\text{Bethe}}\left(\Gamma,\underline{s},\eta\right)=\frac{1}{n}\sum_{i\in V}\ln\left(q\left(s_{i}\mid+1\right)\right).\label{eq:LC_ferro_Bethe}
\end{equation}
The factors entering in the weight function \eqref{eq:LC_loop_weight}
are computed using Equations \eqref{eq:LC_kappa_a} for check nodes
\begin{equation}
\kappa_{a}=\begin{cases}
1 & \left|\partial_{g}a\right|\,\text{is\,even}\\
0 & \left|\partial_{g}a\right|\,\text{is\,odd}
\end{cases},\label{eq:LC_kappa_a_ferro}
\end{equation}
and Equation~\eqref{eq:LC_kappa_i} for variable nodes 
\begin{equation}
\kappa_{i}=\begin{cases}
\left(-1\right)^{l}e^{-2\left(\lambda\left(s_{i}\right)+\eta\right)} & \left|\partial_{g}i\right|=l\\
0 & \left|\partial_{g}i\right|<l
\end{cases},\label{eq:LC_kappa_i_ferro}
\end{equation}
where in the last expression we have used the half log-likelihood
variables 
\begin{equation}
\lambda\left(s_{i}\right):=\frac{1}{2}\ln\frac{q\left(s_{i}\mid+1\right)}{q\left(s_{i}\mid-1\right)}.\label{eq:LC_log_likelihood}
\end{equation}
 Based on the expression of the factors \eqref{eq:LC_kappa_a_ferro}
and \eqref{eq:LC_kappa_i_ferro}, the only subgraphs with a non-zero
weight are those with an induced variable-node degree equal to $l$
and even induced check-node degree. This motivates the definition
of the ferromagnetic loops ensemble 
\begin{equation}
\mathcal{L}_{\Gamma}^{+}=\left\{ g\in\mathcal{L}_{\Gamma}\mid\forall i,a\in g,\left|\partial_{g}i\right|=l\text{ and }\left|\partial_{g}a\right|\,\text{is\,even}\right\} .\label{eq:LC_loop_ferro}
\end{equation}
A loop that is not an element of the ferromagnetic ensemble has a
zero weight. Moreover the weight of a ferromagnetic loop is always
non-negative 
\begin{eqnarray}
K_{+}\left(g\right) & = & \exp\left(-2\eta\left|V_{g}\right|-2\sum_{i\in V_{g}}\lambda\left(s_{i}\right)\right)\geq0.\label{eq:LC_weight_ferro}
\end{eqnarray}
In order to see that $K_{+}\left(g\right)$ is non-negative, notice
that a sign is only associated with the factors $\kappa_{i}$ and
is equal to $\left(-1\right)^{l}$. Therefore a loop can only have
a negative weight if the product $l\left|V_{g}\right|$ is odd. Note
that this product is the number of edges in a loop counted from the
variable-node perspective. Therefore it should be equal to the number
of edges counted from the check-node perspective

\begin{equation}
l\left|V_{g}\right|=\sum_{a\in C_{g}}\left|\partial_{g}a\right|.
\end{equation}
Since for a ferromagnetic loop $\left|\partial_{g}a\right|$ is always
even, $l\left|V_{g}\right|$ is also even and the weight of a loop
is always non-negative.

Using Equations~\eqref{eq:LC_loop_correction} and \eqref{eq:LC_ferro_Bethe}
we can express the average free entropy \eqref{eq:fe_free_entropy}
in the simple form

\begin{eqnarray}
\mathbb{E}_{\Gamma,\underline{s}}\left[\phi\left(\Gamma,\underline{s},\eta\right)\right] & = & \mathbb{E}_{\Gamma,\underline{s}}\left[\frac{1}{n}\ln\left(1+\sum_{g\in\mathcal{L}_{\Gamma}^{+}}K_{+}\left(g\right)\right)\right]\nonumber \\
 & &+  \int dsq\left(s\mid1\right)\ln\left(q\left(s\mid1\right)\right).\label{eq:LC_ferro_Bethe_loop_decomposition}
\end{eqnarray}
Note that Equation~\eqref{eq:LC_ferro_Bethe_loop_decomposition}
is valid for all BMS channels regardless of the noise parameter. However
we can only expect that the ferromagnetic loop-series vanishes in
the decoding regime.

\section{First Moment Method on the Loop Series\label{sec:First-Moment-Method}}

We use a first moment method to prove that the ferromagnetic loop-series
in Equation~\eqref{eq:LC_ferro_Bethe_loop_decomposition} vanishes.
In our case it is based on Jensen's inequality and consists of permuting
the expectation over the graph ensemble and the logarithm in Equation~\eqref{eq:LC_ferro_Bethe_loop_decomposition}. 

Note that we cannot permute the expectation over the channel output
realizations and the logarithm. It is easy to see that over the channel
output realizations a loop has an expected weight \eqref{eq:LC_weight_ferro}
that increases exponentially fast for $\eta<0$ 
\begin{equation}
\mathbb{E}_{\underline{s}}\left[K_{+}\left(g\right)\right]=e^{-\eta\left|V_{g}\right|}.
\end{equation}
This is because the loop series is dominated by events for which most
of the bits are corrupted and have negative half log-likelihood \eqref{eq:LC_log_likelihood}.
These events are rare but give rise to an exponentially large weight.

Therefore we estimate the expectation of the loop series over the
ensemble of regular $\left(l,r\right)$ bipartite graphs for a fixed
output realization of the channel.

\subsection{Probability Estimates on Graphs}

For a given channel realization $\underline{s}$ of the BEC (resp.
BSC) call $V_{c}$ the set of variable nodes with $s_{i}=0$ (resp.
$s_{i}=-1$) and call $V_{0}$ the set of variable nodes $i\in V$
with $s_{i}=1$ (resp. $s_{i}=1$). The set $V_{0}$ contains bits
that have been correctly transmitted and $V_{c}$ contains bits that
have been corrupted. We denote the fraction of correctly transmitted
bits by $\left(1-\rho\right)=\left|V_{0}\right|/n$ and we denote
the fraction of corrupted bits by $\rho=\left|V_{c}\right|/n$. We
recall that the total number of variable nodes is $n=\left|V\right|$
and the total number of check nodes is $m=\left|C\right|$.

We decompose the set of ferromagnetic loops \eqref{eq:LC_loop_ferro}
into subsets of loops having the same ``type''. The type of a loop
$g\in\mathcal{L}_{\Gamma}^{+}$ is the triplet $\left(x_{0},x_{c},\underline{y}\right)\in\left[0,1\right]^{2\times\left\lfloor r/2\right\rfloor }$
where $x_{0}=\left|V_{0}\cap V_{g}\right|/n$ is the fraction of correctly
transmitted variable nodes in the loop, $x_{c}=\left|V_{c}\cap V_{g}\right|/n$
is the fraction of corrupted variable nodes in the loop and $\underline{y}=\left\{ y_{t}\right\} _{t=1}^{\left\lfloor r/2\right\rfloor }$
is the fraction of check nodes with degree $2t$. The set of loops
of type $\left(x_{0},x_{c},\underline{y}\right)$ is denoted by $\Omega\left(x_{0},x_{c},\underline{y}\right)$.

Not all value of $\left(x_{0},x_{c},\underline{y}\right)$ are admissible
loop types. The fraction of check nodes inside a loop is upper bounded
by $1$. Moreover counting edges from the variable-node perspective
or from the check-node perspective obviously gives the same number.
Therefore types that are admissible belong to the following set already
introduced in Section \ref{sec:Main-Results}, Eq. \eqref{eq:pr_domain}
\begin{eqnarray}
D\left(\rho\right) & = & \left\{ \left(x_{0},x_{c},\underline{y}\right)\in\left[0,1\right]^{2+\left\lfloor r/2\right\rfloor }\mid\sum_{t=1}^{\left\lfloor r/2\right\rfloor }y_{t}\leq1,\right.\nonumber \\
 &  & \left.\sum_{t=1}^{\left\lfloor r/2\right\rfloor }\frac{2t}{r}y_{t}=\left(1-\rho\right)x_{0}+\rho x_{c}\right\} .\label{eq:fm_domain}
\end{eqnarray}

The weight \eqref{eq:LC_weight_ferro} of a loop $g\in\Omega\left(x_{0},x_{c},\underline{y}\right)$
is only a function of its type $K_{+}\left(g\right)\equiv K_{+}\left(x_{0},x_{c}\right)$.
Using the specific expression of the half log-likelihood \eqref{eq:LC_log_likelihood}
for each channels we find the explicit form of the weight function
for the BEC
\begin{equation}
K_{+}^{\text{BEC}}\left(x_{0},x_{c}\right)=\begin{cases}
\exp\left(-2n\eta x_{c}\rho\right) & x_{0}=0\\
0 & x_{0}>0
\end{cases},\label{eq:fm_weight_type_BEC}
\end{equation}
and for the BSC 
\begin{eqnarray}
K_{+}^{\text{BSC}}\left(x_{0},x_{c}\right) & = & \exp\left(-2n\eta\left(x_{0}\left(1-\rho\right)+x_{c}\rho\right)\right.\nonumber \\
 &  & \left.+nk\left(x_{0},x_{c},\rho,p\right)\right),\label{eq:fm_weight_type_BSC}
\end{eqnarray}
where $k\left(x_{0},x_{c},\rho,p\right)$ is the auxiliary function
introduced in Section~\ref{sec:Main-Results}, Eq. \eqref{eq:pr_k_function}
\begin{equation}
k\left(x_{0},x_{c},\rho,p\right)=\left(\rho x_{c}-\left(1-\rho\right)x_{0}\right)\ln\left(\frac{1-p}{p}\right).
\end{equation}
Therefore the expected value of the loop series over the graph ensemble
can be expressed only through loop types 
\begin{eqnarray}
\mathbb{E}_{\Gamma}\left[\sum_{g\in\mathcal{L}_{\Gamma}^{+}}K_{+}\left(g\right)\right] & = & \sum_{\left(x_{0},x_{c},\underline{y}\right)\in D\left(\rho\right)}K_{+}\left(x_{0},x_{c}\right)\nonumber \\
 & &\times  \mathbb{E}_{\Gamma}\left[\left|\Omega\left(x_{0},x_{c},\underline{y}\right)\right|\right].\label{eq:fe_average_loop_series}
\end{eqnarray}

The expected number of loops with prescribed type $\left(x_{0},x_{c},\underline{y}\right)$
is upper bounded using McKay's combinatorial estimate\footnote{McKay's bound in its original form is only applicable for subgraphs
of size less than $n-4r^{2}$. We refer to \cite{Macris2013} for
a careful analysis.} \cite{mckay2010subgraphs} for subgraphs with specified degrees

\begin{eqnarray}
\mathbb{E}_{\Gamma}\left[\left|\Omega\left(x_{0},x_{c},\underline{y}\right)\right|\right] & \leq & n^{\delta{}_{l,r}}\left(\begin{array}{c}
nl\\
nl\left(x_{0}\left(1-\rho\right)+x_{c}\rho\right)
\end{array}\right)^{-1}\nonumber \\
 & &\times  \left(\begin{array}{c}
n\left(1-\rho\right)\\
nx_{0}\left(1-\rho\right)
\end{array}\right)\left(\begin{array}{c}
n\rho\\
nx_{c}\rho
\end{array}\right)\nonumber \\
 & &\times  \left(\begin{array}{c}
m\\
my_{1},\ldots,my_{\left\lfloor r/2\right\rfloor }
\end{array}\right)\nonumber \\
 & &\times  \prod_{t=1}^{\left\lfloor r/2\right\rfloor }\left(\begin{array}{c}
r\\
2t
\end{array}\right)^{my_{t}},\label{eq:fe_McKay_estimate}
\end{eqnarray}
where $\delta_{l,r}$ is a constant that depends only on $l$ and
$r$. McKay's estimate has the advantage to have an asymptotically
tight growth rate when $n$ goes to infinity.

It remains to prove that the average loop series \eqref{eq:fe_average_loop_series}
with the bound \eqref{eq:fe_McKay_estimate} vanishes in the large
$n$ limit.

\subsection{Laplace's Method and Proof of Theorems}

The loop series \eqref{eq:fe_average_loop_series} is dominated by
loop types that contribute to the sum with the biggest exponential
growth. We apply Laplace's method in order to characterize the biggest
exponent. 

Using Stirling inequalities
\begin{equation}
e^{\frac{1}{12n+1}}\leq\frac{n!}{\sqrt{2\pi n}e^{-n}n^{n}}\leq e^{\frac{1}{12n}},
\end{equation}
we find an asymptotically tight upper bound on the estimate \eqref{eq:fe_McKay_estimate}
\begin{equation}
\mathbb{E}_{\Gamma}\left[\left|\Omega\left(x_{0},x_{c},\underline{y}\right)\right|\right]\leq C_{l,r}n^{\delta'_{l,r}}\exp\left(nf\left(x_{0},x_{c},\underline{y},\rho\right)\right),\label{eq:fm_laplace_mckay}
\end{equation}
where $C_{l,r}$ and $\delta'_{l,r}$ are just numerical constants
and $f\left(x_{0},x_{c},\underline{y},\rho\right)$ is the auxiliary
function introduced in Section \ref{sec:Main-Results}, Eq. \eqref{eq:pr_partity_check}
\begin{eqnarray}
f\left(x_{0},x_{c},\underline{y},\rho\right) & = & -lh_{2}\left(\left(1-\rho\right)x_{0}+\rho x_{c}\right)\nonumber \\
 & &+  \left(1-\rho\right)h_{2}\left(x_{0}\right)+\rho h_{2}\left(x_{c}\right)\nonumber \\
 & &-  \frac{l}{r}\left(1-\sum_{t=1}^{r}y_{t}\right)\ln\left(1-\sum_{t=1}^{r}y_{t}\right)\nonumber \\
 & &-  \frac{l}{r}\sum_{t=1}^{r}y_{t}\ln y_{t}\nonumber \\
 & &+  \frac{l}{r}\sum_{t=1}^{\left\lfloor r/2\right\rfloor }y_{t}\ln\left(\begin{array}{c}
r\\
2t
\end{array}\right).\label{eq:fm_f_function}
\end{eqnarray}
Combining Equations~\eqref{eq:fm_weight_type_BEC}, \eqref{eq:fm_weight_type_BSC}
and \eqref{eq:fm_laplace_mckay}, we show that the leading exponent
in Equation \eqref{eq:fe_average_loop_series} is for the BEC 
\begin{equation}
\alpha^{\text{BEC}}\left(\rho,\eta\right)=\max_{\left(0,x_{c},\underline{y}\right)\in D\left(\rho\right)}f\left(0,x_{c},\underline{y},\rho\right)-2\eta x_{c}\rho,\label{eq:fm_alpha_BEC}
\end{equation}
and is for the BSC
\begin{eqnarray}
\alpha^{\text{BSC}}\left(\rho,\eta\right) & = & \max_{\left(x_{0},x_{c},\underline{y}\right)\in D\left(\rho\right)}\left(-2\eta\left(x_{0}\left(1-\rho\right)+x_{c}\rho\right)\right.\nonumber \\
 & &+  \left.f\left(x_{0},x_{c},\underline{y},\rho\right)+k\left(x_{0},x_{c},\rho,p\right)\right).\label{eq:fm_alpha_BSC}
\end{eqnarray}
 Notice that for all $\rho$ and $\eta$ the exponent $\alpha^{\text{BEC/BSC}}\left(\rho,\eta\right)$
is non-negative. This is easily verified by evaluating the objective
function at $\left(x_{0},x_{c},\underline{y}\right)=\left(0,0,0\right)$.
Therefore the bit-error probability vanishes if $\alpha^{\text{BEC/BSC}}\left(\rho,\eta\right)$
is equal to zero for all $\eta$ in a neighborhood of zero. The next
Lemma shows that in fact only the maximization at $\eta=0$ is important.
\begin{lem}
\label{lem:fm_eta_zero}If the maximum of \eqref{eq:fm_alpha_BEC}
(resp. \eqref{eq:fm_alpha_BSC}) is uniquely achieved in $\left(x_{0},x_{c},\underline{y}\right)=\left(0,0,0\right)$
for $\eta=0$, then there exists $\widetilde{\eta}<0$ such that $\alpha^{\text{BEC}}\left(\rho,\eta\right)=0$
(resp. $\alpha^{\text{BSC}}\left(\rho,\eta\right)=0$) for all $\eta\in\left]\widetilde{\eta},\infty\right[$. \end{lem}
\begin{proof}
See Appendix \ref{app:Proof-of-Lemma-eta}
\end{proof}
In order to prove Theorems~\ref{thm:PR_concentration_error_BEC}
and \ref{thm:PR_concentration_error_BSC}, we need to show that small
variations around $\rho$ do not change $\alpha^{\text{BEC}}\left(\rho,0\right)$
and $\alpha^{\text{BSC}}\left(\rho,0\right)$. This is guaranteed
by the following Lemma.
\begin{lem}
\label{lem:fm_typical_sequence_optimization} For all $\rho\in\left[0,1\right]$,
if $\alpha^{\text{BEC}}\left(\rho,0\right)=0$ (resp. $\alpha^{\text{BSC}}\left(\rho,0\right)=0$)
and the maximum of \eqref{eq:fm_alpha_BEC} (resp. \eqref{eq:fm_alpha_BSC})
is uniquely achieved at $\left(x_{0},x_{c},\underline{y}\right)=\left(0,0,0\right)$,
there exists $N$ sufficiently large such that 
\[
\forall n\geq N,\,\forall\delta\in\left[-\sqrt{\frac{\ln n}{n}},\sqrt{\frac{\ln n}{n}}\right],\,\alpha^{\text{BEC/BSC}}\left(\rho+\delta,0\right)=0
\]
\end{lem}
\begin{proof}
See Appendix \ref{app:Proof-of-Lemma-typical}
\end{proof}
We are now in position to prove our main theorems.
\begin{proof}[Proof of Theorem \ref{thm:PR_concentration_error_BEC}]
\begin{singlespace}
 Let $\epsilon$ be the probability of error of the BEC. First notice
that the perturbed partition function \eqref{eq:fe_perturbed_partition_function}
is trivially lower bounded by $1$ and upper bounded by $2^{n}e^{2n\left|\eta\right|}$.
This implies that the free entropy \eqref{eq:fe_free_entropy} remains
finite
\begin{equation}
0\leq\phi\left(\Gamma,\underline{s},\eta\right)\leq\ln2+2\left|\eta\right|.
\end{equation}
Therefore using Equation~\eqref{eq:LC_ferro_Bethe_loop_decomposition}
and the fact that $K_{+}\left(g\right)\geq0$ we see that the loop
series remains finite as well 
\begin{eqnarray}
2\left(\ln2+\left|\eta\right|\right) & \geq & \left|\mathbb{E}_{\underline{s}}\left[\phi\left(\Gamma,\underline{s},\eta\right)\right]\right.
\nonumber \\
&& - \left.\int dsq\left(s\mid1\right)\ln\left(q\left(s\mid1\right)\right)\right|\nonumber \\
 & = & \mathbb{E}_{\underline{s}}\left[\frac{1}{n}\ln\left(1+\sum_{g\in\mathcal{L}_{\Gamma}^{+}}K_{+}\left(g\right)\right)\right].\label{eq:fm_loop_series_bound}
\end{eqnarray}
\end{singlespace}

Let $A$ be the following probabilistic event on the channel output
realizations
\begin{equation}
A:=\left\{ \underline{s}\in\left\{ -1,0,1\right\} ^{n}\mid\left|\frac{1}{n}\sum_{i=1}^{n}s_{i}-(1-\epsilon)\right|\leq\sqrt{\frac{\ln n}{n}}\right\} .
\end{equation}
Output realizations in $A$ are close to the average output realization.

Using Hoeffding's inequality, we see that the probability of the complementary
event $A^{c}$ vanishes 
\begin{equation}
\mathbb{P}_{\underline{s}}\left[A^{c}\right]\leq\frac{2}{n^{-2}}.
\end{equation}
Combining Jensen's inequality and the trivial bound \eqref{eq:fm_loop_series_bound}
on the loop series we have the following estimate
\begin{eqnarray}
\mathbb{E}_{\Gamma,\underline{s}}\left[\frac{1}{n}\ln\left(1+\sum_{g\in\mathcal{L}_{\Gamma}^{+}}K_{+}\left(g\right)\right)\right]\leq\frac{4}{n^{-2}}\left(\ln2+\left|\eta\right|\right)\nonumber \\
+\mathbb{E}_{\underline{s}}\left[\frac{1}{n}\ln\left(1+\mathbb{E}_{\Gamma}\left[\sum_{g\in\mathcal{L}_{\Gamma}^{+}}K_{+}\left(g\right)\right]\right)\mid A\right].\label{eq:fm_conditionned_expectation}
\end{eqnarray}
Since we have conditioned over channel output realizations that are
in $A$, the fraction of corrupted bit is $\left|\rho-\epsilon\right|\leq\sqrt{\ln n/n}$.
Therefore combining Equation~\eqref{eq:fm_laplace_mckay}, Lemma~\ref{lem:fm_eta_zero}
and Lemma~\ref{lem:fm_typical_sequence_optimization} we have that
if $\alpha^{\text{BEC}}\left(\epsilon,0\right)=0$ is uniquely achieved
in $\left(x_{0},x_{c},\underline{y}\right)=\left(0,0,0\right)$ then
for all $\eta\in\left]\widetilde{\eta},\infty\right[$ and $n$ sufficiently
large,
\begin{eqnarray}
\mathbb{E}_{\underline{s}}\left[\frac{1}{n}\ln\left(1+\mathbb{E}_{\Gamma}\left[\sum_{g\in\mathcal{L}_{\Gamma}^{+}}K_{+}\left(g\right)\right]\right)\mid A\right]\leq\nonumber \\
\frac{1}{n}\ln\left(1+c_{3}n^{c_{4}}\right),
\end{eqnarray}
where $c_{3}$ and $c_{4}$ are numerical constants independent of
$n$.

We have proved that for all $\eta\in\left]\widetilde{\eta},\infty\right[$
with $\widetilde{\eta}<0$ the average free entropy converges in expectation
over the regular $\left(l,r\right)$ LDPC ensemble 

\begin{equation}
\lim_{n\rightarrow\infty}\mathbb{E}_{\Gamma}\left[\left|\mathbb{E}_{\underline{s}}\left[\phi\left(\Gamma,\underline{s},\eta\right)\right]-\int dsq\left(s\mid1\right)\ln\left(q\left(s\mid1\right)\right)\right|\right]=0.
\end{equation}
In particular it implies that the average free entropy over the LDPC
ensemble converges 
\begin{equation}
\lim_{n\rightarrow\infty}\mathbb{E}_{\Gamma,\underline{s}}\left[\phi\left(\Gamma,\underline{s},\eta\right)\right]=\int dsq\left(s\mid1\right)\ln\left(q\left(s\mid1\right)\right).
\end{equation}
Since $\mathbb{E}_{\Gamma,\underline{s}}\left[\phi\left(\Gamma,\underline{s},\eta\right)\right]$
is a convex function of $\eta$ and converges pointwise in a neighborhood
of zero, we can exchange the limit and the derivative
\begin{eqnarray}
0 & = & \left.\frac{\partial}{\partial\eta}\lim_{n\rightarrow\infty}\mathbb{E}_{\Gamma,\underline{s}}\left[\phi\left(\Gamma,\underline{s},\eta\right)\right]\right|_{\eta=0}\nonumber \\
 & = & \lim_{n\rightarrow\infty}\left.\frac{\partial}{\partial\eta}\mathbb{E}_{\Gamma,\underline{s}}\left[\phi\left(\Gamma,\underline{s},\eta\right)\right]\right|_{\eta=0}\nonumber \\
 & = & \lim_{n\rightarrow\infty}\left.\mathbb{E}_{\Gamma}\left[\frac{\partial}{\partial\eta}\mathbb{E}_{\underline{s}}\left[\phi\left(\Gamma,\underline{s},\eta\right)\right]\right]\right|_{\eta=0}\nonumber \\
 & = & -2\lim_{n\rightarrow\infty}\mathbb{E}_{\Gamma}\left[P_{\Gamma}^{\text{bit-sampling}}\right],
\end{eqnarray}
where in the last line we use Equation~\eqref{eq:fe_free_entropy_bit_error}
that relates the free entropy to the bit-error probability.
\end{proof}
Theorem~\ref{thm:PR_concentration_error_BSC} has a proof almost
identical to that of Theorem~\ref{thm:PR_concentration_error_BEC}.

\section{Path Forward\label{sec:Path-Forward}}

We would like to stress that the techniques developed in this paper
are quite general. In particular they do not rely on a special form
of channels or on the regular-degree distribution of the LDPC ensemble.
Therefore we plan to improve our results in the following ways.

\subsection{Generalization to Arbitrary Degree Distributions}

The entire analysis can easily be extended to general degree distributions
with bounded degrees. It will simply transform the function \eqref{eq:fm_f_function}
that counts subgraphs into a more convoluted object. However extending
our results to distributions with unbounded degrees, like for instance
Poisson distributions, may be more complicated. One would have to
derive an estimate for counting subgraphs in this particular case.

\subsection{Asymmetric Channels}

The loop series and the Bethe free entropy for general channels are
almost exactly similar than for symmetric channels. For general channels
we can no longer assume that the all-zero codeword is transmitted.
Instead we have to average the bit-error probability over all possible
input codewords $\underline{\tau}$. In this case the weight of a
loop remains similar than for symmetric channels. The weight is also
non-negative and depends on the generalized half log-likelihood ratio
\begin{equation}
\lambda\left(s_{i}\mid\tau_{i}\right)=\frac{1}{2}\log\frac{q\left(s_{i}\mid\tau_{i}\right)}{q\left(s_{i}\mid-\tau_{i}\right)},
\end{equation}
where $\underline{s}$ denotes as usual the channel observations.
In order to control the loop series, we will need to perform a conditioned
expectation in \eqref{eq:fm_conditionned_expectation} over joint
typical sequences of input codewords and noise realizations.

\subsection{Tight Thresholds}

As described in Section \ref{sec:Main-Results}, the thresholds that
we obtain are not tight. In fact at fixed rate they become worse and
converge to zero as the degrees of the graph become large. The reason
why we obtain such loose bounds for large degrees comes from the function
$f\left(x_{0},x_{c},\underline{y},\rho\right)$ defined in \eqref{eq:fm_f_function}.
This function counts the growing rate of the average number of subgraphs
with a prescribed type $\left(x_{0},x_{c},\underline{y}\right)$ 
\begin{equation}
f=\lim_{n\rightarrow\infty}\frac{1}{n}\ln\left(\mathbb{E}_{\Gamma}\left[\left|\Omega\left(x_{0},x_{c},\underline{y}\right)\right|\right]\right).\label{eq:pf_f}
\end{equation}
One can verify that if instead of $f$ we use the function
\begin{equation}
\widetilde{f}=\lim_{n\rightarrow\infty}\frac{1}{n}\mathbb{E}_{\Gamma}\left[\ln\left(\left|\Omega\left(x_{0},x_{c},\underline{y}\right)\right|\right)\right],\label{eq:pf_f_tilde}
\end{equation}
we obtain tight lower and upper bound on the threshold for vanishing
bit-error probability. 

The function $\widetilde{f}$ only depends on the random graph ensembles
that we consider and does not depend on a particular channel. Computing
this function would provide a proof of the exact location of the MAP
threshold for an extensive class of channels. However this computation
could prove to be a very difficult task.

A way around the problem of computing \eqref{eq:pf_f_tilde} is to
condition the expectation \eqref{eq:pf_f} on some rare events with
respect to the random graph measure. Note that by Jensen's inequality
$\widetilde{f}$ is always upper-bounded by $f$. This is because
the expectation \eqref{eq:pf_f} is dominated by rare events that
are associated with a large weight $\left|\Omega\left(x_{0},x_{c},\underline{y}\right)\right|$.
Conditioning on these rare events will lead to better estimates of
\eqref{eq:pf_f_tilde} and will provide tighter bounds at least in
the limit of large degrees.

\appendices{}

\section{Proof of Lemma~\ref{lem:fm_eta_zero}\label{app:Proof-of-Lemma-eta}}
\begin{proof}
We prove Lemma~\ref{lem:fm_eta_zero} only for the BSC (the proof
for the BEC is almost identical). For a given $\rho$ and $p$, let
us define the following function
\begin{eqnarray}
g^{\text{BSC}}\left(x_{0},x_{c},\underline{y},\eta\right) & = & f\left(x_{0},x_{c},\underline{y},\rho\right)+k\left(x_{0},x_{c},\rho,p\right)\nonumber \\
 & &-  2\eta\left(x_{0}\left(1-\rho\right)+x_{c}\rho\right).\label{eq:app_g_BSC}
\end{eqnarray}
The function $g^{\text{BSC}}\left(x_{0},x_{c},\underline{y},\eta\right)$
corresponds to the exponent of the loop series \eqref{eq:fe_average_loop_series}
associated with the loop type $(x_{0},x_{c},\underline{y})$. In order
to prove Lemma~\ref{lem:fm_eta_zero}, we have to find $\widetilde{\eta}<0$
such that $g^{\text{BSC}}$ is non-positive on $D(\rho)\times\left[\widetilde{\eta},+\infty\right[$. 

We first show that for any $\widetilde{\eta}_{1}<0$, there exists
a neighborhood $U$ of $(x_{o},x_{c},\underline{y})=(0,0,0)$ such
that $g^{\text{BSC}}$ is non-positive on $U\cap D(\rho)\times\left[\widetilde{\eta}_{1},+\infty\right[$.
For a fixed $\widetilde{\eta}_{1}<0$ we construct a function $\overline{g}^{\text{BSC}}$
that is an upper bound of $g^{\text{BSC}}$. We restrict ourselves
to the domain $V\cap D(\rho)\times\left[\widetilde{\eta}_{1},+\infty\right[$,
where $V=\mathbb{B}(0,1/3r)$ is the ball of radius $1/3r$ centered
at $(0,0,0)$.

Let us explicitly write down the function \eqref{eq:app_g_BSC} term
by term

\begin{eqnarray}
g^{\text{BSC}}\left(x_{0},x_{c},\underline{y},\eta\right) & = & -2\eta\left(x_{0}\left(1-\rho\right)+x_{c}\rho\right)\nonumber \\
 & &+  \left(\rho x_{c}-\left(1-\rho\right)x_{0}\right)\ln\left(\frac{1-p}{p}\right)\nonumber \\
 & &+  \frac{l}{r}\sum_{t=1}^{\left\lfloor r/2\right\rfloor }y_{t}\ln\left(\begin{array}{c}
r\\
2t
\end{array}\right)\nonumber \\
 & &-  lh_{2}\left(\left(1-\rho\right)x_{0}+\rho x_{c}\right)\nonumber \\
 & &+  \left(1-\rho\right)h_{2}\left(x_{0}\right)+\rho h_{2}\left(x_{c}\right)\nonumber \\
 & &-  \frac{l}{r}\left(1-\sum_{t=1}^{r}y_{t}\right)\ln\left(1-\sum_{t=1}^{r}y_{t}\right)\nonumber \\
 & &-  \frac{l}{r}\sum_{t=1}^{r}y_{t}\ln y_{t}.\label{eq:app_g_function}
\end{eqnarray}
We bound each term of \eqref{eq:app_g_function} separately. Denote
the fraction of variable nodes in the loop by $X=x_{0}\left(1-\rho\right)+x_{c}\rho$.
The inequalities below trivially hold 
\begin{eqnarray}
\left(\rho x_{c}-\left(1-\rho\right)x_{0}\right)\ln\left(\frac{1-p}{p}\right) & \leq & 2\ln\left(\frac{1-p}{p}\right)X\nonumber \\
\frac{l}{r}\sum_{t=1}^{\left\lfloor r/2\right\rfloor }y_{t}\ln\left(\begin{array}{c}
r\\
2t
\end{array}\right) & \leq & l\ln\left(\begin{array}{c}
r\\
2\left\lfloor r/2\right\rfloor 
\end{array}\right)X\nonumber \\
-2\eta\left(x_{0}\left(1-\rho\right)+x_{c}\rho\right) & \leq & -2\widetilde{\eta}_{1}X.\label{eq:app_upb_1}
\end{eqnarray}
As the entropy is a concave function, we have the following inequality
\begin{equation}
\left(1-\rho\right)h_{2}\left(x_{0}\right)+\rho h_{2}\left(x_{c}\right)\leq h_{2}(X).\label{eq:app_upb_2}
\end{equation}
Concativty of $-x\ln x$ gives us
\begin{eqnarray}
-\sum_{t=1}^{\left\lfloor r/2\right\rfloor }y_{t}\ln y_{t} & \leq & -\left(\sum_{t=1}^{\left\lfloor r/2\right\rfloor }y_{t}\right)\ln\left(\frac{1}{\left\lfloor r/2\right\rfloor }\sum_{t=1}^{\left\lfloor r/2\right\rfloor }y_{t}\right)\nonumber \\
 & \leq & -\left(\sum_{t=1}^{\left\lfloor r/2\right\rfloor }y_{t}\right)\ln\left(\sum_{t=1}^{\left\lfloor r/2\right\rfloor }y_{t}\right)\nonumber \\
 &  & +r\ln\left(\left\lfloor r/2\right\rfloor \right)X.\label{eq:app_upb_3}
\end{eqnarray}
Note that since the domain is restricted to types in a ball of radius
$1/3r$, the fraction of variable nodes in a loop is upper-bounded
$X\leq1/3r$. In particular it implies that 
\begin{eqnarray}
\sum_{t=1}^{\left\lfloor r/2\right\rfloor }y_{t} & \leq & \frac{r}{2}\left(\sum_{t=1}^{\left\lfloor r/2\right\rfloor }\frac{2t}{r}y_{t}\right)\nonumber \\
 & = & \frac{r}{2}X\nonumber \\
 & \leq & \frac{1}{6}\nonumber \\
 & \leq & \frac{1}{e},
\end{eqnarray}
where $e$ is the Euler constant. Finally as the entropy is increasing
on $\left[0,\frac{1}{e}\right]$, we have
\begin{eqnarray}
\frac{l}{r}h_{2}\left(\sum_{t=1}^{\left\lfloor r/2\right\rfloor }y_{t}\right) & \leq & \frac{l}{r}h_{2}\left(\frac{r}{2}X\right).\label{eq:app_upb_4}
\end{eqnarray}

The upper bound on the function \eqref{eq:app_g_function} is simply
the sum of Inequalities \eqref{eq:app_upb_1}, \eqref{eq:app_upb_2},
\eqref{eq:app_upb_3} and depends only on the fraction of variable
nodes in a loop i.e. $\overline{g}^{\text{BSC}}\left(x_{0},x_{c},\underline{y},\eta\right)\equiv\overline{g}^{\text{BSC}}(X)$ and 
\begin{eqnarray}
\overline{g}^{\text{BSC}}(X) & = & \frac{l}{r}h_{2}\left(\frac{r}{2}X\right)-(l-1)h_{2}\left(X\right)+MX,
\end{eqnarray}
where $M$ is a constant independent of $\eta$ and $\rho$ 
\begin{equation}
M=2\ln\left(\frac{1-p}{p}\right)+l\ln\left(\begin{array}{c}
r\\
2\left\lfloor r/2\right\rfloor 
\end{array}\right)+l\ln\left(\left\lfloor r/2\right\rfloor \right)-2\widetilde{\eta}_{1}.
\end{equation}
Notice that $\overline{g}^{\text{BSC}}(0)=0$ and that the derivative
$\frac{d}{dX}\overline{g}^{\text{BSC}}(X)$ behaves like $\left(\frac{l}{2}-1\right)\ln X$
in the neighborhood of $0$. Hence, for $l\geq3$, there exists $\delta>0$
such that $\overline{g}^{\text{BSC}}$ is negative on $\left]0,\delta\right]$.
Therefore for all types $(x_{o},x_{c},\underline{y})\in D(\rho)$
in the domain $U=\mathbb{B}(0,\delta)\cap\mathbb{B}(0,1/3r)$ and
for all $\eta\in\left[\widetilde{\eta}_{1},+\infty\right[$ we have

\begin{eqnarray}
g^{\text{BSC}}\left(x_{0},x_{c},\underline{y},\eta\right) & \leq & \overline{g}^{\text{BSC}}(X)\nonumber \\
 & \leq & 0.
\end{eqnarray}

By hypothesis the maximum of \eqref{eq:fm_alpha_BSC} is uniquely
achieved in $(0,0,0)$ for $\eta=0$. It implies that there exists
$\lambda<0$ such that 
\begin{equation}
\max_{\left(x_{0},x_{c},\underline{y}\right)\in D\left(\rho\right)\setminus U}f\left(x_{0},x_{c},\underline{y},\rho\right)+k\left(x_{0},x_{c},\rho,p\right)=\lambda.\label{eq:maxawayzero}
\end{equation}
Therefore for $\eta>\widetilde{\eta}_{2}=\lambda/2$ 
\begin{equation}
\max_{\left(x_{0},x_{c},\underline{y}\right)\in D\left(\rho\right)\setminus U}g^{\text{BSC}}\left(x_{0},x_{c},\underline{y},\eta\right)\leq\lambda-2\widetilde{\eta}_{2}=0.
\end{equation}

We see that $\widetilde{\eta}=\max(\widetilde{\eta}_{1},\widetilde{\eta}_{2})<0$
satisfies by construction the condition of Lemma~\ref{lem:fm_eta_zero}.
\end{proof}

\section{Proof of Lemma~\ref{lem:fm_typical_sequence_optimization}\label{app:Proof-of-Lemma-typical}}
\begin{proof}
We prove Lemma~\ref{lem:fm_typical_sequence_optimization} only for
the BSC (the proof for the BEC is almost identical). For a given $\rho$
and $p$, we recall the function $g_{p,\rho}^{\text{BSC}}\equiv g^{\text{BSC}}$
and $\overline{g}_{p,\rho}^{\text{BSC}}\equiv\overline{g}^{\text{BSC}}$
as defined in Appendix \ref{app:Proof-of-Lemma-eta}. We prove that
for $n$ sufficiently large and for all $\delta\in\left[-\sqrt{n^{-1}\ln n},\sqrt{n^{-1}\ln n}\right]$,
the function $g_{p,\rho+\delta}^{\text{BSC}}$ is still non-positive
on $D(\rho)$. 

First notice that the upper bound $\overline{g}_{p,\rho}^{\text{BSC}}$
does not depend on $\rho$. Using the same argument as in Appendix
\ref{app:Proof-of-Lemma-eta}, there exists a neighborhood $U$ of
$(0,0,0)$ such that for all type $(x_{0},x_{c},\underline{y})\in U\cap D(\rho+\delta)$
and for all $\delta\in\left[-\sqrt{n^{-1}\ln n},\sqrt{n^{-1}\ln n}\right]$
\begin{eqnarray}
g_{p,\rho+\delta}^{\text{BSC}}\left(x_{0},x_{c},\underline{y},0\right) & \leq & \overline{g}_{p,\rho}^{\text{BSC}}(X)\nonumber \\
 & \leq & 0.
\end{eqnarray}

It remains to show that the variation of $g^{\text{BSC}}$ on $D\left(\rho+\delta\right)\setminus U$
is bounded. Let us make the change of variables $\left(x_{0},x_{c}\right)\rightarrow\left(X,x_{c}\right)$
and $g_{p,\rho+\delta}^{\text{BSC}}\left(x_{0},x_{c},\underline{y},0\right)\rightarrow g_{p,\rho+\delta}^{\text{BSC}}\left(X,x_{c},\underline{y},0\right)$.
The following inequality holds
\begin{eqnarray}
g_{p,\rho+\delta}^{\text{BSC}}\left(X,x_{c},\underline{y},0\right) & \leq & 2\sqrt{\frac{\ln n}{n}}\left(\ln2+\ln\left(\frac{1-p}{p}\right)\right)\nonumber \\
 & &+  g_{p,\rho}^{\text{BSC}}\left(X,x_{c},\underline{y},0\right).
\end{eqnarray}
Hence we can bound the maximum of $g^{\text{BSC}}$ on $D(\rho+\delta)\setminus U$
by 
\begin{eqnarray}
\max_{(X,x_{c},\underline{y})\in D(\rho+\delta)\setminus U} & g_{p,\rho+\delta}^{\text{BSC}}\left(X,x_{c},\underline{y},0\right) & \leq\nonumber \\
\max_{(X,x_{c},\underline{y})\in D(\rho)\setminus U} & g_{p,\rho}^{\text{BSC}}\left(X,x_{c},\underline{y},0\right) & +c\sqrt{\frac{\ln n}{n}}.
\end{eqnarray}
The maximum of $g_{p,\rho}^{\text{BSC}}\left(X,x_{c},\underline{y},0\right)$
on $D(\rho+\delta)\setminus U$ is by hypothesis negative (see Equation
\eqref{eq:maxawayzero}). Therefore for $n$ sufficiently large we
have that for all $\delta\in\left[-\sqrt{n^{-1}\ln n},\sqrt{n^{-1}\ln n}\right]$

\begin{eqnarray}
\max_{(x_{0},x_{c},\underline{y})\in D(\rho+\delta)\setminus U}g_{p,\rho+\delta}^{\text{BSC}}\left(x_{0},x_{c},\underline{y},0\right) & \leq & 0,
\end{eqnarray}
 which concludes the proof.
\end{proof}

\section*{Acknowledgment}

M. Vuffray acknowledge support from the LDRD Program at Los Alamos
National Laboratory by the National Nuclear Security Administration
of the U.S. Department of Energy under Contract No. DE-AC52-06NA25396.
The work of T. Misiakiewicz was partially supported by the National
Science Foundation award No. 1128501 at New Mexico Consortium. The
authors thanks N. Macris and M. Chertkov for discussions and comments.

\bibliographystyle{IEEEtran}
\bibliography{refloop,references}

\end{document}